\documentclass[12pt]{article}
\usepackage{algorithm}
\usepackage{algorithmic}

\usepackage[pdftex]{graphicx}
\usepackage{subfigure}
\usepackage{sidecap}
\usepackage{epsfig,amsthm,amsmath,color, amsfonts}
\usepackage{epsfig,color}

\usepackage{framed}
\newtheorem{theorem}{Theorem}[section]
\newtheorem{corollary}[theorem]{Corollary}

\newtheorem{lemma}[theorem]{Lemma}

\theoremstyle{definition}
\theoremstyle{definition}\newtheorem{definition}[theorem]{Definition}
\theoremstyle{observation}

\newcommand{\comment}[1]{}
\newcommand{\QED}{\mbox{}\hfill \rule{3pt}{8pt}\vspace{10pt}\par}
\def\eps{{\epsilon}}

\def\polylog{\operatorname{polylog}}
\newcommand{\ignore}[1]{}
\newcommand{\eat}[1]{}

\usepackage{times}
\usepackage[small,compact]{titlesec}
\usepackage[small,it]{caption}

\newcommand{\squishlist}{
 \begin{list}{$\bullet$}
  { \setlength{\itemsep}{0pt}
     \setlength{\parsep}{3pt}
     \setlength{\topsep}{3pt}
     \setlength{\partopsep}{0pt}
     \setlength{\leftmargin}{1.5em}
     \setlength{\labelwidth}{1em}
     \setlength{\labelsep}{0.5em} } }
\newcommand{\squishend}{
  \end{list}  }

\def\eps{{\epsilon}}

\date{}

\begin{document}
\title{Distributed Computation of Sparse Cuts}
%\begin{titlepage}
\author{Atish {Das Sarma} \thanks{eBay Research Labs, eBay Inc., CA, USA.
\hbox{E-mail}:~{\tt atish.dassarma@gmail.com.}}\and  Anisur Rahaman Molla \thanks{Division of Mathematical
Sciences, Nanyang Technological University, Singapore 637371. \hbox{E-mail}:~{\tt anisurpm@gmail.com}.} \and Gopal Pandurangan \thanks{Division of Mathematical
Sciences, Nanyang Technological University, Singapore 637371 and Department of Computer Science, Brown University, Providence, RI 02912, USA. \hbox{E-mail}:~{\tt gopalpandurangan@gmail.com}.  Supported in part by the following research grants: Nanyang Technological University grant M58110000, Singapore Ministry of Education (MOE) Academic Research Fund (AcRF) Tier 2 grant MOE2010-T2-2-082, MOE  AcRF Tier 1 grant MOE2012-T1-001-094, and a grant from the US-Israel Binational Science Foundation (BSF).}}

\date{}

\maketitle \thispagestyle{empty}

\maketitle

\begin{abstract}
Finding  sparse cuts is an important tool in analyzing  large-scale distributed networks such as the Internet and Peer-to-Peer networks, as well as large-scale graphs such as the web graph, online social communities, and VLSI circuits. Sparse cuts are useful in graph clustering and partitioning among numerous other applications. In distributed communication networks, they are useful for topology maintenance and  for designing better search and routing algorithms.

In this paper, we focus on developing fast distributed algorithms for computing sparse cuts in networks. Given an undirected $n$-node network $G$ with conductance $\phi$, the goal is to find a cut set whose conductance is close to $\phi$. We present two distributed algorithms that find a cut set with sparsity $\tilde O(\sqrt{\phi})$ ($\tilde{O}$ hides $\polylog{n}$ factors). Both our algorithms work in the CONGEST distributed computing model and output a cut of conductance at most $\tilde O(\sqrt{\phi})$ with high probability, in $\tilde O(\frac{1}{b}(\frac{1}{\phi} + n))$ rounds, where $b$ is balance of the cut of given conductance. In particular, to find a sparse cut of constant balance, our algorithms take $\tilde O(\frac{1}{\phi} + n)$ rounds.
 
Our  algorithms can also be used to output a {\em local} cluster, i.e., a subset of vertices near a given source node, and whose conductance is within a quadratic factor of the best possible cluster around the specified node. Both our distributed algorithm can work without knowledge of the optimal $\phi$ value and hence can be used to find approximate conductance values both globally and with respect to a given source node. We also give a lower bound on the time needed for any distributed algorithm to compute any non-trivial sparse cut --- any distributed approximation algorithm (for any non-trivial  approximation ratio) for computing sparsest cut will take $\tilde \Omega(\sqrt{n} + D)$ rounds, where $D$ is the diameter of the graph.

Our algorithm can be used to find  sparse cuts 
(and their conductance values) and to identify well-connected clusters and critical edges in distributed networks. This in turn can be helpful in the design, analysis, and maintenance of topologically-aware networks.

\end{abstract}


\medskip

\section{Introduction}\label{sec:intro}
Developing distributed algorithms
for computing key metrics of a communication network is an important research goal with
various applications. 
Network properties 
--- which depend on the collective behavior of  nodes
and links ---  characterize global network  performance
such as  routing, sampling, information dissemination, etc.  These in turn depend on topological properties  of the network such as high connectivity, low diameter, high conductance, and good spectral properties \cite{mihail}.
%For example, overlay P2P  networks, which are virtual networks built over the Internet, are used for file sharing and content distribution applications.
%A P2P  network's topology should have good connectivity properties  for providing
%good quality of service at the virtual network
%layer, as well as for providing load balancing at the underlying network layer.
%Furthermore,
%it might be desirable that network is highly connected which enables robust communication
%even under failures. 
The above properties, all of which are critical, need to be measured periodically. Having a highly-connected network is good for fault-tolerance and reliable routing,
since a packet can be routed via many disjoint paths. Low diameter ensures 
that packets can be routed quickly with short delay. Conductance  (formally defined in Section \ref{sec:def}) measures how  
``well-knit" the network is; it determines how fast a random walk 
converges to the stationary distribution --- known as the {\em mixing time}.
 Conductance is related to the {\em expansion}, {\em spectral gap}, and mixing time of a graph.  High expansion and spectral gap means that the graph
has fast mixing time. Such a network supports fast random sampling which
has many applications \cite{drw-jacm} and low-congestion routing \cite{mihail}. 
%The spectral properties (the graph spectrum --- the set of eigenvalues of the adjacency matrix) tell a great deal of the network structure. 

Sparse cuts are those cuts that have low conductance and can be used
to determine well-connected clusters\footnote{A cut $(S,V-S)$ is a partition
of the set of nodes $V$ into $S$ (assume $|S| \leq |V|/2$) and $V-S$. A low conductance cut has lot more edges within $S$ than those going outside $S$, and hence $S$ is relatively well (intra)connected.} and thus also identify potential ``bottlenecks" in the network. In particular, the edges crossing the cut
can be considered as {\em critical} edges and they have been used in designing algorithms to
improve searching, topology maintenance (i.e., maintaining a well-connected topology), and reducing routing congestion in networks \cite{mihail}.
%Such algorithms are useful in the design, analysis, and maintenance of {\em topology aware} networks \cite{mihail}.
  
In this paper, we focus on developing fast distributed algorithms for computing sparse cuts  in networks. 
Given an undirected $n$-node network $G$ with conductance $\phi$ (a quantity less than 1), the goal is to find a cut set whose conductance is close to $\phi$. (We note that computing the minimum conductance cut --- the one with conductance $\phi$ of the network--- is NP-hard \cite{MatulaS90}.)  We present two distributed algorithms that find a  cut set with sparsity $\tilde O(\sqrt{\phi})$. Both our algorithms use small-sized messages and work in the CONGEST distributed computing model. Our algorithms build on previous work \cite{LovaszS90,SpielmanT04,AndersenCL06} on classical algorithms for sparse cuts. 

Both algorithms output a cut of conductance at most $\tilde O(\sqrt{\phi})$ with high probability, in $\tilde O(\frac{1}{b}( \frac{1}{\phi} + n))$ rounds, where $b$ is balance of the cut of given conductance (cf. Section \ref{sec:def}). In particular, to find a cut of constant balance (i.e., the cuts are of  approximately equal size), the first algorithm takes $\tilde O(\frac{1}{\phi} + n)$ rounds and finds such a cut (if it exists) with similar approximation.  The second  
 algorithm is a variant of the first one and based on a different approach involving PageRank. Our  algorithm can also be used to output a well-connected {\em local} cluster (cf. Section \ref{sec:def}) , i.e., a subset  $S$ of vertices containing the given source node such that the internal edge connections in $S$ are significantly higher than the outgoing edges from $S$. 
Both our distributed algorithms can work  without knowledge of the optimal $\phi$ value and hence can be used
to find approximate conductance values both globally and locally with respect to a given source node. We also show a lower bound
on the time needed for any distributed algorithm to compute any non-trivial sparse cut. In particular, we show that there is graph in which any distributed approximation algorithm (for any non-trivial  approximation ratio, not just quadratic approximation) for computing sparsest cut will take $\tilde \Omega(\sqrt{n} + D)$ rounds, where $D$ is the diameter of the graph. 

%It is known  (see Footnote \ref{foot:cheeger}) that $\tilde O(\sqrt{M}) \leq \frac{1}{\phi} \leq \tilde O(M)$ and hence the second algorithm in general   
%is at least as fast as the first one and can be up to quadratic times faster.

Our algorithm can be useful in efficiently finding sparse cuts 
(and their conductance values) and critical edges (the edges crossing sparse cuts) in distributed networks.
In particular, the work of \cite{mihail} shows how  critical edges can be used to design algorithms to  improve search, reduce congestion in routing, and for keeping the graph well-connected (topology maintenance).
Such  algorithms can be useful in the design and deployment
of {\em reconfigurable networks} (whose topology can be changed by rewiring edges) such as  peer-to-peer networks and  wireless mesh networks. The paper \cite{KerenS12} study information spreading where they used a generalized notion of conductance as a key tool. In fact, the conductance helps to identify bottlenecks in the network and thus achieves fast information spreading.   

%, which in turn can be helpful in the design, analysis, and maintenance
%of {\em topologically-(self)aware} networks, i.e.,  networks that can monitor and regulate themselves in a decentralized fashion \cite{mihail}.

The focus of distributed computation of spectral properties that we are interested
here, in particular, conductance and sparse cuts, is relatively new. 
The work of \cite{drw-jacm} presented a fast decentralized algorithm for estimating mixing time, conductance, and spectral gap of the network.
% In
%particular, they show that given a starting point $x$, the mixing time with respect to $x$, called $\tau^x_{mix}$, can be
%estimated in $\tilde{O}(n^{1/2} + n^{1/4}\sqrt{D\tau^x_{mix}})$ rounds, where $D$ is the network diameter. 
%If the estimate of $\tau^x_{mix}$ is close to the mixing time of the network defined as $\tau_{mix} = \max_{x}{\tau^x_{mix}}$, then this allows one to estimate also the conductance $\phi$ (upto a quadratic factor)
%and spectral gap of the graph\footnote{\label{foot:cheeger} The spectral gap is the $1-\lambda_2$ where $\lambda_2$ is the second eigenvalue of the connected transition matrix. It is known that conductance, mixing time, and spectral gap are related to each other \cite{JS89}:  $\frac{1}{1-\lambda_2}\leq \tau_{mix}\leq \frac{\log n}{1-\lambda_2}$ and $\Theta(1-\lambda_2)\leq \Phi\leq \Theta(\sqrt{1-\lambda_2})$.}. 
The work of Kempe and McSherry \cite{kempe} gives a decentralized algorithm 
for computing the top eigenvectors of a weighted adjacency matrix
that runs in $O(\tau_{mix}\log^2 n)$ round, where $\tau_{mix}$ is the mixing time of the network\footnote{Estimating mixing time  also allows one to estimate conductance $\phi$ (upto a quadratic factor) and spectral gap of the graph. The spectral gap is  $1-\lambda_2$ where $\lambda_2$ is the second eigenvalue of the connected transition matrix. It is known that conductance, mixing time, and spectral gap are related to each other \cite{JS89}:  $\frac{1}{1-\lambda_2}\leq \tau_{mix}\leq \frac{\log n}{1-\lambda_2}$ and $\Theta(1-\lambda_2)\leq \Phi\leq \Theta(\sqrt{1-\lambda_2})$.}.

While the above works give distributed algorithms to estimate the conductance $\phi$, they {\em do not}
give an efficient distributed algorithm  to compute sparse cuts. Sparse cuts have low conductance (i.e., close to $\phi$) and, in particular, the sparsest cut
is a cut that achieves the network conductance. Since there are exponential number of cuts in the network, it is significantly more challenging to efficiently find
the sparsest cut or approximate it in a distributed fashion. Hence computing sparse cuts needs
a different approach compared to computing conductances and mixing time as in the works of \cite{drw-jacm, kempe}.

Our approach, on a high-level, is based on efficiently implementing
 the methods of Lov{\'a}sz and Simonovits \cite{LovaszS90,SpielmanT04,AndersenCL06}.
  This method uses random walks to
 estimate the probability distribution of such walks terminating at nodes.
 This probability distribution can then be used to identify sparse cuts.
 Our first algorithm uses standard random walks (cf. Section \ref{sec:sparse-cut}). Our second algorithm uses
 random walks with {\em reset} to a given source node, in other words, it computes {\em personalized PageRank} (cf. Section \ref{sec:pagerank-algo}). 
% Our algorithms can be used to estimate the  conductance of the network.
% The second algorithm, in particula
%as well as ``local" conductance, i.e., conductance of a sparse set
%containing a given source node.

\iffalse
Our approach crucially use random walks. Random walks
are very  local and lightweight and requires little index or state maintenance 
that makes it attractive to self-organizing networks such as Internet overlay and ad hoc wireless networks \cite{zhong,BBSB04}. 
 A subroutine for efficiently sampling from several random walks naturally leads to a technique
for estimating the spectral gap, or the mixing time of the network graph in a distributed manner.
Samples from walks of length $\ell$ are samples from the distribution induced at length $\ell$. These algorithms
can be extended to sample from walks of length $1, 2, 4, 8,\dots$. One can then find an efficient way to
compare the distance between distributions and estimate $t$ such that the distribution at $t$ and $2t$ are
close. This gives an estimate of the mixing time and thereby the spectral gap. 
\fi

\subsection{Model and Definitions}\label{sec:model}
\subsubsection{Distributed Computing Model}
\label{sec:distmodel}
We model the communication network as an undirected, unweighted\footnote{We restrict our attention on unweighted graphs for the upper bound analysis, however, our algorithm can be extended to weighted graphs as well.}, connected $n$-node graph $G = (V, E)$. Every  node has limited initial knowledge. Specifically, assume that each node is associated with a distinct identity number  (e.g., its IP address). 
%Here for simplicity, we assume that the identity numbers are from the set $\{1, 2, . . . , n\}$ . 
At the beginning of the computation, each node $v$ accepts as input its own identity number and the identity numbers of its neighbors in $G$. The node may also accept some additional inputs as specified by the problem at hand. The nodes are allowed to communicate through the edges of the graph $G$. We assume that the communication occurs in  synchronous  {\em rounds}. 
%In particular, all the nodes wake up simultaneously at the beginning of round 1, and from this point on the nodes always know the number of the current round. 
We will use only small-sized messages. In particular, in each round, each node $v$ is allowed to send a message of size $O(\log n)$ through each edge $e = (v, u)$ that is adjacent to $v$.  The message  will arrive to $u$ at the end of the current round. 
%This is a standard model of distributed computation known as the {\em CONGEST model} \cite{peleg} and has been attracting a lot of research attention during last two decades (e.g., see \cite{peleg} and the references therein).
This is a  widely used  standard model known as the {\em CONGEST model} to study distributed algorithms (e.g., see \cite{peleg, PK09}) and captures the bandwidth constraints inherent in real-world computer  networks. 
%Our algorithms can be easily generalized if $B$ bits  are allowed (for any pre-specified parameter $B$) to be sent through each edge in a round. Typically, as assumed here, $B = O(\log n)$, which is number of bits needed to send a node id in a n-node network. 
(We note that if unbounded-size messages were allowed through every edge in each time step, then the problem addressed here can be trivially solved in $O(D)$ time by collecting all the topological information at one node, solving the problem locally, and then broadcasting the
results back to all the nodes \cite{peleg}.) 

There are several measures of efficiency of distributed algorithms, but we will focus on one of them, specifically, {\em the running time}, i.e. the number of {\em rounds} of distributed communication. Note that the computation that is performed by the nodes locally is ``free'', i.e., it does not affect the number of rounds; however, we will only perform polynomial cost computation locally in any node. We note that in the CONGEST model, it is rather trivial to solve a
problem in $O(m)$ rounds, where $m$ is the number of edges in the network, since
the entire topology (all the edges) can be collected at one node and the problem solved locally. The goal is to design faster algorithms.

\subsubsection{Definitions}
\label{sec:def}
We present notations that we use throughout the paper.
Consider a graph $G = (V, E)$ with conductance $\phi$ and let $|V| = n, |E| = m$. Let $p_{\ell}(s, t)$ denote the probability that a  random walk of length $\ell$ starting from $s$ ends in $t$. In fact, $p_{\ell}(s,t)$ is the probability distribution over the nodes after a walk of length $\ell$ starting from $s$. We simply use $p(t)$ instead of $p_{\ell}(s, t)$ when source node and length is clear from the text. Let $S$ be a subset of $V$. We denote a partition or cut by $(S, \bar{S})$ or sometimes by $(S, V\setminus S)$ interchangeably throughout the paper. For a probability distribution $p(i)$ on nodes, let $\rho_p(i) = p(i)/d(i)$. Let $\pi_p$ denote the ordering of nodes in decreasing order of $\rho_p(i)$.

\begin{definition}[{\bf Conductance and Sparsity}]
\label{def:conductance}
The conductance of a cut $(S, \bar{S})$ (also called as sparsity) is $\phi(S) = \frac{|E(S, \bar{S})|}{\min \{\mbox{vol}(S), 2m - \mbox{vol}(S)\}}$ where $\mbox{vol}(S)$ is the sum of the degrees of nodes in $S$. The conductance of the graph $G$ is $\phi(G) = \min_{S \subseteq V} \phi(S)$. We denote it by only $\phi$, if it is clear from the text.  
\end{definition}
%\vspace{-0.05in}
\begin{definition}[{\bf Balance}]
\label{def:balance}
The balance of a cut $(S, \bar{S})$ is defined as $\min \{\frac{|S|}{|V|}, \frac{|\bar{S}|}{|V|} \}$ and is denoted by $b$. 
\end{definition}

\begin{definition}[{\bf Local Cluster}]
\label{def:balance}
A local cluster with respect to a given vertex $v$ is a subset $S \subset V$ containing  $v$ such that the conductance of $(S, \bar{S})$ is within a quadratic factor of the best possible local cluster containing  $v$. %In other words, the internal edge connections in $S$ are significantly higher than the outgoing edges from $S$.      
\end{definition}

%arxiv
\subsection{Problem Statement and Our Results}\label{sec:results}

\subsubsection{Problem Statement}
In this paper, we consider the problem of finding a sparse cut in an undirected graph. Formally, given a graph $G = (V, E)$ with conductance $\phi$, we want to find a cut set whose conductance is close to $\phi$. Our goal is to design a distributed algorithm which finds a cut set with sparsity $\tilde O(\sqrt{\phi})$.   

\subsubsection{Our Results}
Our main contributions are two distributed algorithms in the CONGEST model to find sparse cuts with approximation guarantees.
Both our algorithms crucially use random walks.

\begin{theorem}(cf. Section \ref{sec:sparse-cut})
\label{thm:algo1}
Given an $n$-node network $G$ with a cut of balance $b$ and conductance at most $\phi$, there is a distributed algorithm {\sc SparseCut} (cf. Algorithm \ref{alg:sparsecut}) that outputs a cut of conductance at most $\tilde O(\sqrt{\phi})$ with high probability, in $\tilde O(\frac{1}{b}(\frac{1}{\phi} + n))$ rounds. In particular, to find a cut of constant balance, the {\sc SparseCut} algorithm takes $\tilde O(\frac{1}{\phi} + n)$ rounds and finds a cut (if it exists) with similar approximation.  
\end{theorem}

The second algorithm is a variant of the first algorithm and based on a PageRank-based approach. The second algorithm achieves the similar running time bound as above.\\ 

Using the above results, we also show:

\begin{theorem}\label{thm:cluster}
Given an $n$-node network $G$ and  source node $s$, there is a  distributed algorithm that  outputs a {\em local} cluster in $\tilde O(\frac{1}{\phi} + n)$ rounds, where $\phi$ is the conductance of the graph. 
\end{theorem}

To prove the above running time bound, we derive a technical result on computing  conductances of $n$ (different) cuts in linear time (cf. Lemma \ref{lem:parallel-conductance}).  

\noindent We note that the time bound of  $\tilde O(\frac{1}{\phi} + n)$  is linear in $n$ (the number of nodes) and $1/\phi$.  From the definition of conductance (cf. Definition \ref{def:conductance}), it is clear that for every graph,
$1/\phi = O(m)$ ($m$ is the number of edges) and for many graphs it can be much smaller, e.g., for expanders it is $O(1)$. Hence, the running time of our algorithms can be significantly faster than the naive bound of $O(m)$ (cf. Section \ref{sec:distmodel}), especially in well-connected dense graphs. We next show a lower bound on the time needed for any distributed algorithm to compute a (non-trivial) sparse cut.

\begin{theorem}(cf. Section \ref{sec:lower-bound})
\label{thm:lb}
There is a $n$-node graph in which any distributed approximation algorithm  for computing sparsest cut (within any non-trivial approximation ratio)  will take $\tilde \Omega(\sqrt{n} + D)$ rounds, where $D$ is the diameter of the graph.
\end{theorem}

Since $1/\phi = \Omega(D)$ for any graph, the above lower bound says that in general, one cannot hope to improve on the $1/\phi$ term
of our upper bound.

\subsection{Outline of This Chapter}
The next two section developes the two different approach to compute sparse cuts. In Section \ref{sec:sparse-cut}, we present the standard random walk-based distributed algorithm for sparse cut problem, by introducing the main ideas. Section \ref{sec:local-cluster} describes on finding local cluster set. Then in Section \ref{sec:pagerank-algo}, we present the second algorithm using PageRank-based approach. Section \ref{sec:lower-bound} derive a general lower bound to the sparse cut computation problem. Finally, we conclude in Section \ref{sec:conclusion} by summarizing the results developed in this chapter and discuss some open problems.

% !TEX root = ipdps-main.tex
\subsection{Related Work}\label{sec:related}

The problem of finding sparse cuts on graphs has been studied extensively
\cite{BhattL84,BenczurK96,Karger00,AroraRV04,SpielmanT04,ManokaranNRS08,DasSarmaGP09}. 
Sparse cuts form an important tool for analyzing large-scale distributed networks such as the Internet and Peer-to-Peer networks, as well as large-scale graphs such as the web graph, online social communities, click graphs from search engine query logs and VLSI circuits. Sparse cuts are useful in graph clustering and partitioning among numerous other applications \cite{SpielmanT04,AndersenCL06}.

The second eigenvector of the transition matrix is an important quantity to analyze many properties of a graph. A simple way of graph partitioning is by ordering the nodes in increasing order of coordinate values in the eigenvector. This partition can be used to compute sparse cut. This is a well known approach studied in \cite{LovaszS90,LovaszS93,SpielmanT04,AndersenCL06,DasSarmaGP09}. We use this approach in this paper. The second eigenvector technique has been analyzed in many papers \cite{Alon86,Boppana87,JerrumS88}.  

Lov{\'a}sz and Simonovits \cite{LovaszS90,LovaszS93} first show how random walks can be used to find sparse cuts. Specifically,
they show that random walks of length $O(1/\phi)$ can be used to compute a cut with sparsity at most $\tilde O(\sqrt{\phi})$ if the sparsest cut has conductance $\phi$. 
Spielman and Teng \cite{SpielmanT04} mostly follow the work of Lov{\'a}sz and Simonovits, but they implement it more efficiently by sparsifying the graph. They propose a nearly linear time algorithm for finding an approximate sparsest cut with approximate balance. 
Andersen, Chung, and Lang \cite{AndersenCL06} proposed a local partitioning algorithm using PageRank vector (instead of second eigenvector) to find cuts
near a specified vertex and global cuts. The running time of their algorithm was proportional
to the size of small side of the cut. Das Sarma, Gollapudi and Panigrahy \cite{DasSarmaGP09} present an algorithm for finding sparse cut in graph streams. Their algorithm requires sub-linear
space for a certain range of parameters, but provides much a weaker approximation to the sparsest cut compared to \cite{AndersenCL06,SpielmanT04}.
Arora, Rao, and Vazirani \cite{AroraRV04} provide $O(\sqrt{\log n})$-approximation algorithm
using semi-definite programming techniques. Their algorithm gives good approximation ratio, however it is slower than algorithms based on spectral methods and
random walks. Kannan, Vempala, and Vetta \cite{KannanVV04} studied variants of spectral algorithm for clustering or partitioning a graph. 

Graph partitioning or rather clustering is an well studied optimization problem. Suppose we are given an undirected graph and a conductance parameter $\phi$. The problem of finding a partition $(S, \bar{S})$ such that $\phi(S) \leq \phi$, or conclude no such partition exits is NP-complete problem (see, \cite{LeightonR99},\cite{SimaS06}). As a result, several approximation algorithms exits in literature. Leighton and Rao presents $O(\log n)$ approximation of the sparsest cut algorithm in \cite{LeightonR99} where they used linear programming. Later Arora, Rao, and Vazirani \cite{AroraRV04} improved this to $O(\sqrt{\log n})$ using semi-definite programming techniques. This is the best known approximation of the sparsest cut computation problem. %Anisur: IS IT TRUE? 
Further, several works obtains algorithm with similar approximation guarantees algorithm but better running time such as \cite{AroraHK04}, \cite{KhandekarRV06}, \cite{AroraK07}, \cite{OrecchiaSVV08}. However, unfortunately no work have been found in distributed computing model. Our paper is the first to attempt in distributed setting for sparse cuts computation.  

The work of \cite{mihail} discusses spectral algorithms for enhancing the {\em topology awareness}, e.g., by identifying and assigning weights to {\em critical} edges of the network.  Critical edges are those that cross sparse cuts. 
 They discuss centralized
algorithms with provable performance, and introduce
decentralized heuristics with no provable guarantees. These algorithms are
based on distributed solutions of convex programs  and assign special weights to links
crossing or directed towards small cuts by minimizing
the second eigenvalue.
It is mentioned that obtaining provably efficient decentralized algorithms is an important open problem.
Our algorithms are fully
decentralized and  based on performing random walks, and so more
amenable to dynamic and self-organizing networks.

\section{A Distributed Algorithm for Sparse Cut}\label{sec:sparse-cut}
%\subsection{Algorithm for Sparse Cut}\label{sec:sparsy}
In this section, we present an algorithm to find a cut that approximates the minimum conductance $\phi$. We are given a network,  $G = (V, E)$, that has cut of conductance $\phi$ and balance $b$. We design a distributed algorithm running on $G$ to compute a cut set $S$ with conductance  $\tilde O (\sqrt{\phi})$. At the end of the algorithm, every node in $G$ will know whether it is in $S$ or $\bar{S}$. Further, each node will also know all other nodes in $S$ or $\bar{S}$. Our algorithm works in the standard CONGEST model of distributed computing (cf. Section \ref{sec:model}).    Without loss of generality, we will assume that our algorithm 
knows $\phi$ and $b$. Otherwise, we can do the following. Suppose we want to find a cut with the required sparsity, i.e., $\tilde O (\sqrt{\phi})$, without knowing
$\phi$, but assume that we know $b$ (the balance of such a cut). Then we can guess the value of $\phi$ starting from a constant (say $1/2$, which is essentially
the highest possible) and then run our algorithm and check whether the output cut value satisfies the quadratic factor approximation (and the given balance). If yes, we stop; otherwise,
we halve our guess and continue. If we don't know $b$ as well,   then our algorithm  (with some assumed balance) will still work and will give a cut with similar quadratic approximation to the minimum conductance cut that is minimum among all possible  cuts with the assumed balance. Thus, henceforth we will assume that our algorithm knows both $\phi$ and $b$.

The outline of our approach  (cf. Theorem \ref{thm:conductance-estimate}) is to try several different cuts obtained by various distributions of random walks. Further these distributions need to be computed from a {\em good} source node. A good source node is one from the smaller side of the desired cut. In this approach, instead of computing the exact distribution after the chosen length of walk, it suffices to have an approximate distribution of sufficiently high accuracy. Assuming that a good source is used, one needs to estimate the distribution after doing a random walk of length $\ell$ that is sampled uniformly in the range of $\{1, 2, \ldots, O(1/\phi)\}$. For the sampled length $\ell$, estimate the landing probability $p(i)$ at every node $i$. Assume the estimation is $\tilde p(i)$. Then arrange the nodes according to decreasing order of $\rho_{\tilde p} = \tilde p(i)/d(i)$. Suppose the order is $\pi_{\tilde p} = (1, 2,\ldots, n)$. Then, with constant probability, at least one of the $n$ cuts $(S_j, \bar{S}_j)$ has the given conductance (approximated), where $S_j = \{1, 2, \ldots, j\}$. This algorithm and its proof of correctness was given in Spielman and Teng \cite{SpielmanT04}.  To get the required cut with high probability, we run our algorithm for $\Theta(\log n)$ different lengths $\ell$, each chosen independently and uniformly at random in the range of $\{1, 2, \ldots, O(1/\phi)\}$.  For a particular $\ell$, there are $n-1$-partitions and so $n-1$ different conductances. The minimum conductance cut among all the $\Theta(n \log n)$ cuts would be the output of our algorithm. Before going to the main algorithm {\sc SparseCut}, we first present an algorithm to estimate the probability distribution $\tilde p(i)$ of $p(i)$ using random walks.  

\subsection{Estimating Random Walk Probability Distribution}\label{sec:prob-estimate} %$p_{\ell}(s, i)$} 
We focus on estimating $p_{\ell}(s, i)$ which is the probability of landing at node $i$ after a random walk of length $\ell$ from a specific source node $s$. As we noted above, we denote it by simply $p(i)$. The basic idea is to perform several random walks of length $\ell$ from $s$ and at the end, each node $i$ computes the fraction of walks that land at node $i$. It is easy to see that the accuracy of estimation is dependent on the number of random walks that are performed from $s$. Let us parameterize the number as $K$. We show (cf. Lemma \ref{lem:time-randomwalk}) that we can perform a polynomial in $n$ number of random walks without any congestion in the network. We first present the algorithm {\sc EstimateProbability}, and then describe the result on accuracy of the estimation (cf. Lemma \ref{lem:probability-accuracy}). The pseudocode of the algorithm {\sc EstimateProbability} is given below in Algorithm \ref{alg:randomwalk}. 
 
\newcommand{\mindegree}[0]{\delta}
\begin{algorithm}[H]
\caption{\sc EstimateProbability}
\label{alg:randomwalk}
\textbf{Input:} Starting node $s$, length $\ell$, and number of walks $K$.\\% = \Theta(n \log n/\eps)$.\\
\textbf{Output:} $\tilde p(i)$ for each node $i$, which is an estimate of $p(i)$ with explicit bound on additive error.\\
\begin{algorithmic}[1]
%\STATE Each node $t$ maintains a counter number $\eta_t$ to count the number of walks land over it. 

\STATE  Node $s$ creates $K$ tokens of random walks and performs them simultaneously for $\ell$ steps as follows. 

\FOR{each round from $1$ to $\ell$}   

\STATE A node holding random walk tokens, samples a random neighbor corresponding to each token and subsequently sends the appropriate {\em count} to each neighbor. (Note that tokens do not contain any node IDs.)  
%\COMMENT{$M$ is the mixing time of the graph.}
\ENDFOR

\STATE Each node $i$ counts the number of tokens that landed on it --- let this count be $\eta_i$.   

\STATE Each node estimates the probability $\tilde p(i)$ as $\frac{\eta_i}{K}$.

%\STATE Each node $t$ outputs $\tilde p(i)$.

\end{algorithmic}

\end{algorithm}

We show that for $K = \Theta(n^2 \log n/\eps^2)$, the algorithm {\sc EstimateProbability} (cf. Algorithm \ref{alg:randomwalk}) gives an estimation of $p(i)$ with accuracy $p(i) \pm \eps/n$ for each node $i$. In other words, by performing $\Theta(n^2 \log n/\eps^2)$ random walks, if $\tilde p(i)$ is an estimation for $p(i)$, then  $|\tilde p(i) - p(i)| \leq \eps/n$. This follows directly from the following lemma. %Due to space limit, we place the proof of following two lemmas in Appendix.  
\begin{lemma}\label{lem:probability-accuracy}
If the probability of an event $X$ occurring is $p$, then in $t = 4 n^2 \log n/\eps^2$ trials , the fraction of times the event $X$ occurs is $p \pm \frac{\eps}{n}$ with high probability. 
\end{lemma}
\begin{proof}
The proof is  follows from a Chernoff bound: $$ \Pr \left[\frac{1}{t} \sum_{i=1}^t X_i < (1 - \delta)p \right] < \left(\frac{e^{-\delta}}{(1-\delta)^{(1-\delta)}} \right)^{tp} < e^{-tp\delta^2/2}$$ and 
$$\Pr \left[\frac{1}{t} \sum_{i=1}^t X_i > (1 + \delta)p \right] < \left(\frac{e^{\delta}}{(1+ \delta)^{(1+ \delta)}} \right)^{tp}.$$ Where $X_1, X_2, \ldots, X_t$ are $t$ independent identically distributed $0-1$ random variables such that $\Pr[X_i = 1] = p$ and $\Pr[X_i = 0] = (1-p)$. The right hand side of the upper tail bound further reduces to $2^{-\delta t p}$ for $\delta > 2e -1$ and for $\delta <2e - 1$, it reduces to $e^{-tp\delta^2/4}$. 

Let us choose $t = 4n^2\log n/\eps^2$, and $\delta =  \frac{\eps}{pn}$. Consider two cases, when $pn \leq \eps$ and when $pn > \eps$. When $pn \leq \eps$ , the lower tail bound automatically holds as $pn - \eps < 0$. In this case, $\delta > 1$, so we consider the weaker bound of the upper tail bound which is $2^{- \delta t p}$. We get $2^{- \delta t p} = 2^{- \eps t/n} = 2^{- 4 n \log n/\eps} = \frac{1}{n^{(4n/\eps)}}$. Now consider the case when $pn > \eps$. Here, $\delta < 1$ is small and hence the lower and upper tail bounds are $e^{-tp\delta^2/2}$ and $e^{-tp\delta^2/4}$. Therefore, between these two, we go with the weaker bound of $e^{-tp\delta^2/4} = e^{- \frac{tp \eps^2}{4p^2n^2}} = e^{- \frac{1}{p}\log n} = 1/n^{\Theta(1)}$. 
\end{proof}

\begin{lemma}\label{lem:time-randomwalk}
Algorithm {\sc EstimateProbability} (cf. Algorithm \ref{alg:randomwalk}) finishes in $O(\ell)$ rounds, if the number of walks $K$ is at most polynomial in $n$.   
\end{lemma}
\begin{proof}
To prove this, we first show that there is no congestion in the network if we perform at most a polynomial number of random walks from $s$. This follows from the algorithm that each node only needs to count the number of random walk tokens that end on it. Therefore nodes do not need to know from which source node or rather from where it receives the random walk
tokens. Hence it is not needed to send the ID of the source node with the token. Since we consider CONGEST model, a polynomial in $n$ number of token's
count (i.e., we can send count of up to a polynomial number) can be sent in one
message through each edge without any congestion. Therefore, one round is enough to perform one step of random walk for all $K$ walks in parallel, where $K$ is at most polynomial in $n$. This implies that $K$ random walks of length $\ell$ can be performed in $O(\ell)$ rounds. Hence the lemma.
\end{proof}

%\paragraph{Note:} It follows from the above approach that one can perform at most a polynomial number of random walks in parallel. Therefore, one can get much better approximation to estimate $p(i)$ in $O(\ell)$ rounds.  

\subsection{Computation of Sparse Cut}

With the probability approximation result (cf. Lemma \ref{lem:probability-accuracy}) and results from the algorithm {\sc Nibble} in \cite{SpielmanT04}, a key technical result follows (stated below). The result guarantees that one of the cuts formed by $n$-prefixes in a specific sorted order of the probability distribution $\tilde p(i)$ has sparsity $\tilde O(\sqrt{\phi})$ \cite{LovaszS93, SpielmanT04}.

\begin{theorem}\label{thm:conductance-estimate}
Let $(U, \bar{U})$ be a cut of conductance at most $\phi$ such that $|U| \leq |V|/2$. Let $\tilde p(i)$ be an estimate for the probability $p(i)$ of a random walk of length $\ell$ from a source node $s$ from $U$. Assume that $|\tilde p(i) - p(i)| \leq \eps(\sqrt{p(i)/n} + 1/n)$, where $\eps \leq o(\phi)$. Consider the $n-1$ candidate cuts obtained by ordering the vertices in decreasing order of $\rho_{\tilde p}$; each candidate cut $(S_j, \bar{S}_j)$ is obtained by setting $S_j$ equal to the set $(1, 2,\ldots,j)$. If the source node is randomly chosen from $U$ and the length is chosen randomly in the range $\{1, 2, \ldots, O(1/\phi)\}$, then with constant probability, one of these $n-1$ candidate cuts has conductance at most $\tilde O(\sqrt{\phi})$, i.e. $\phi(S_j) \leq \tilde O(\sqrt{\phi})$.    
\end{theorem}
\begin{proof}
The proof is shown in \cite{DasSarmaGP09} and  is implicit in \cite{SpielmanT04} and uses a random walk mixing result from \cite{LovaszS93}. %This result is also derive in  using results from \cite{SpielmanT04}.     
\end{proof}
Therefore, it follows from the above Theorem \ref{thm:conductance-estimate} that if we can estimate the probability $p(i)$ in such a way that it satisfies all the conditions as stated, then  we can find a cut with sparsity $\tilde O(\sqrt{\phi})$. We see that the algorithm {\sc EstimateProbability} estimates $p(i)$ and the error bound is given in Lemma $\ref{lem:probability-accuracy}$. By setting $\eps$ appropriately (which is $O(\phi^2)$), we can satisfy the requirement of Theorem $\ref{thm:conductance-estimate}$. We only need to choose the source node $s$, a bit carefully. The source node $s$ should be sampled from the smaller side of the cut of given conductance $\phi$ as it is required in Theorem $\ref{thm:conductance-estimate}$. But we do not have any idea  about the cut. To overcome this, we sample several source nodes from $V$ and execute the algorithm for every source node.  By sampling $\log n/b$ random nodes from $V$ gives at least one node  from the smaller side of the cut with high probability. Notice that if $b$ is constant then it is enough to choose $O(\log n)$ source nodes.

In the following lemma, we show that one can compute the conductances of $n$ cuts, obtained according to some ordering of vertices, in linear time. In particular, in this paper we use the ordering of the vertices in decreasing order of $\rho_{\tilde p}$. 

%Anisur: I am not sure if this is a right place to put the following lemma----Atish Please check. 

\begin{lemma}[\textbf{$n-1$-Cuts' Conductance}]
\label{lem:parallel-conductance}
Let $G = (V, E)$ be an undirected graph. Let $\pi = (1, 2, \ldots, n)$ be an ordering of $n$ vertices of $G$. Then computing conductances of all $n-1$ cuts $(S_j, \bar{S}_j), j = 1, 2, \ldots, n-1$, can be done in $O(n)$ rounds where $S_j = \{1, 2, \ldots, j\}$.   
\end{lemma}
\begin{proof}
Let us assume that each node in the graph knows the ordering $\pi$, i.e., each node knows its position in the ordering $\pi$. We know from definition of conductance (cf. Definition \ref{def:conductance}) that only  two values are needed, namely $|E(S, \bar{S})|$ (number of crossing edges between $S$ and $\bar{S}$) and $vol(S)$ (sum of degrees of nodes in $S$) to compute the conductance of a cut $(S, \bar{S})$. Therefore, our goal is to collect these two pieces of {\em information} of all the cuts at node $1$ and compute conductances locally. We assume that node $1$ knows $m$, the number of edges in the graph, otherwise, it can be known easily using $O(D)$ rounds by building a breadth-first tree (e.g., after leader election). Notice that the partitions $(S_j, \bar{S}_j)$ are formed by adding nodes one by one from the ordered set $\{1, 2, \ldots, n\}$ starting from the set $S_1 = \{1\}$. Suppose node $1$ has the information of $L^{\pi}_j =$ number of neighbors in $S_{j-1}$ and $R^{\pi}_j =$ number of neighbors in $\bar{S}_j$ (assuming $S_0 = $NULL) for all nodes $j = 1, 2, \ldots, n$. Then, node $1$ can easily compute the value of $|E(S_j, \bar{S}_j)|$ and $vol(S_j)$ for all partitions locally as follows: $|E(S_j, \bar{S}_j)| = |E(S_{j-1}, \bar{S}_{j-1})| - L^{\pi}_j + R^{\pi}_j$ and vol$(S_j) = $ vol$(S_{j-1}) + L^{\pi}_j + R^{\pi}_j$ and $|E(S_1, \bar{S}_1)| =$ deg$(1)$ and vol$(1) = $ deg$(1)$ (where $deg(1)$ is node 1's degree). Therefore, node $1$ starts computing from $S_1, S_2$, and so on up to $S_{n}$.  Note that $L^{\pi}_j + R^{\pi}_j = $ degree of the node $j$. We next mention how node $1$ can have all these information in linear time. This is easy: a node can compute its neighbor's position (i.e., whether its neighbor is in $S_{j-1}$ or in $\bar{S}_j$) in the ordered set $\pi$ in constant number of rounds (cf. Figure \ref{fig:n-cut}). Every node can do this computation  in parallel. It will take constant number of rounds for every node to compute $L^{\pi}_j$ and $R^{\pi}_j$. Then each node $j$ sends the information which contains its ID, $L^{\pi}_j$ and $R^{\pi}_j$ to node $1$ by upcast \cite{peleg}. This will take at most $O(n+D)$ rounds. Then node $1$ can compute conductances locally as discussed above. Therefore, total time taken is $O(n + D)$ rounds to compute all $n$ conductances. This is actually $O(n)$ rounds since $D =O(n)$. 
\end{proof}

\begin{SCfigure}[]
\centering 
\includegraphics[width=0.22\textwidth]{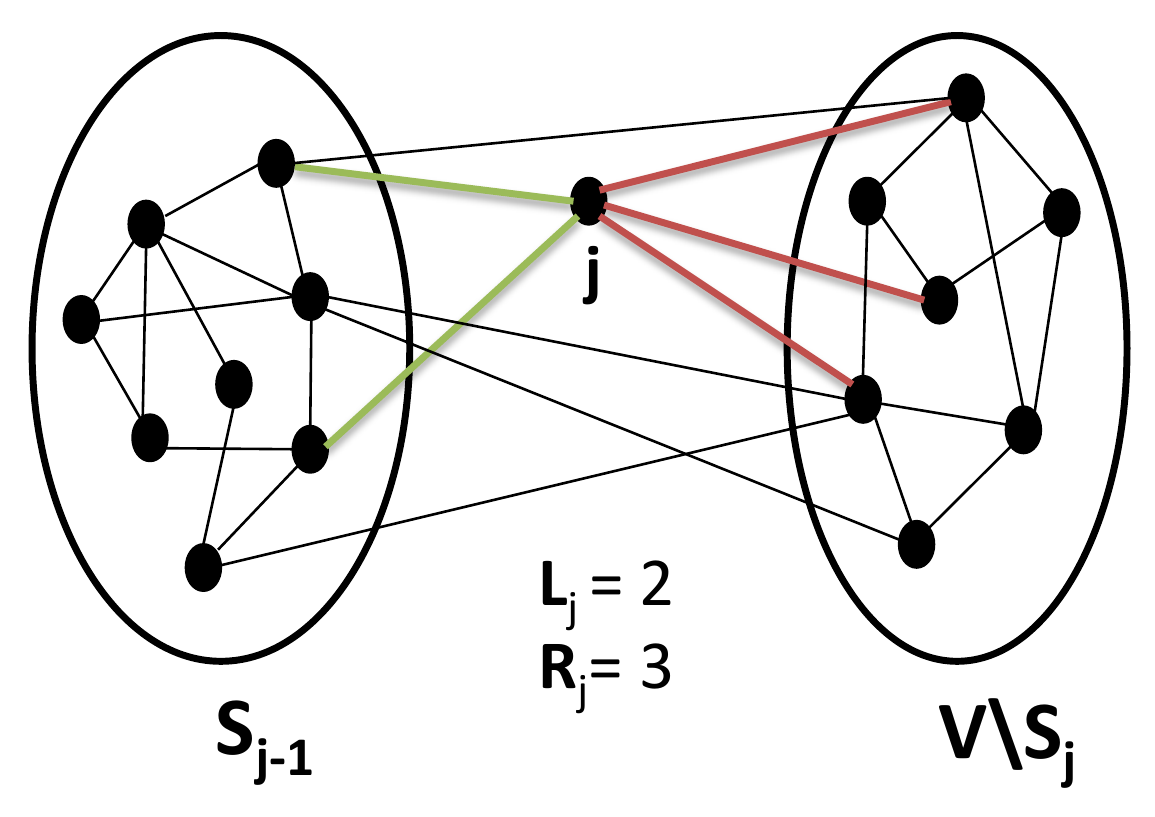}
\caption{Node $j$ computes the number of its neighbors that are in the left side and right side of $j$ in the ordered vertex set $\pi$.}
\label{fig:n-cut}
\end{SCfigure}

The algorithm and description of each step is given below. Complete pseudocode is given in Algorithm \ref{alg:sparsecut} which computes an approximate cut. At the end of our algorithm, each node knows the cut set which has sparsity $\tilde O(\sqrt{\phi})$. %We assume the mixing time $M$ of the graph is an input, as we can compute mixing time of the graph in the beginning of the algorithm using Algorithm \ref{alg:mixing-time}.        

\begin{algorithm}[!h]
%\small
\caption{\sc SparseCut}
\label{alg:sparsecut}
\textbf{Input:} Graph $G = (V, E)$, a conductance $\phi$ of a cut and balance $b$ of the cut (as mentioned in the beginning of Section \ref{sec:sparse-cut}, we assume knowledge of $\phi$ and $b$, without loss of generality).\\
\textbf{Output:} A sparse cut $C = (S, \bar{S})$ with conductance at most $\tilde O(\sqrt{\phi})$.\\

\begin{algorithmic}[1]

\FOR{$k = 1$ to $\log n/b$}
\STATE Choose a source node $s_k$ uniformly at random from $V$. 

%\STATE Source node $s_k$ calls algorithm {\sc EstimateMixingTime} to compute mixing time for source $s_k$. 

%\COMMENT{If the mixing time is an input then we do not need the above step.}

\FOR{$h = 1$ to $\log n$}
\STATE Choose a length $\ell$ uniformly at random in the range of $\{1, 2, \ldots, O(1/\phi)\}$. 

%\end{algorithmic}

\textbf{\{Phase 1: Finding partition of nodes using probability distribution by random walks.\}}

%\begin{algorithmic}%[2]

%\STATE  Choose $\log n$ source nodes $s$ randomly from $V$. Each node choose himself as a source node with probability $\log n/n$. 

\STATE Source node $s_k$ calls algorithm {\sc EstimateProbability} with input $\ell$ and $K = \Theta(\frac{n^2\log n}{\eps^2})$ to compute $p(i)$ for all nodes $i$.  

%\STATE Every node $t$ estimate probabilities $p(t) = p_{\ell}(s, t)$ using {\sc EstimateProbability} algorithm (cf. Algorithm~\ref{alg:randomwalk}) for source node $s$ and walk-length $\ell$. Let  $\tilde p(t)$ be the estimation of  $p(t)$. 

\STATE Each node sends the value $\rho(i) = \tilde p(i)/d(i)$ to all other nodes in the network. 

\STATE Let (without loss of generality) $\pi_{\tilde p} = \{1, 2, \ldots, n\}$ be the ordering of nodes in decreasing order of the set $\{\rho(i) : i \in V\}$. Each node knows $\pi_{\tilde p}$.

%\STATE The algorithm computes the conductances of each $(n-1)$ cuts $(S_j, V\setminus S_j)$, obtained by taking prefixes of the ordering in $\pi_p$, i.e. $S_j = \{1, 2,\ldots, j\}$ for $j=1, 2, \ldots, n-1$, in sequential manner as follows. Note that each node knows, in which partition they lies since they all knows $\pi_p$. %$S_j$ of the $j^{th}$ cut $(S_j, V\setminus S_j)$. 

%\end{algorithmic}

\textbf{\{Phase 2: Finding conductance of the cuts $(S_j,  \bar{S}_j)$ where $S_j = \{1, 2,\ldots, j\}$ for $j=1, 2, \ldots, n-1$ in $\pi_{\tilde p}$.\}}
%\begin{algorithmic}[1]

\STATE Consider node $1$ as master node which collects information of all $n-1$ cuts $(S_j, \bar{S}_j)$ one by one and computes conductances locally as follows.%Assume node $1$ knows the number of edges $m$ in the network, otherwise, this can be known using $O(D)$ rounds by aggregation algorithm. 

%\FOR{each node $j = 1, 2, \ldots, n-1$}

\STATE Every node $j$ does in parallel: compute $L^{\pi}_j =$ number of neighbors in $S_{j-1}$ and $R^{\pi}_j =$ number of neighbors in $\bar{S}_j$ (assuming $S_0 = $NULL). %Notice that  $L^{\pi}_j + R^{\pi}_j = $ degree of the node $j$.

%\STATE Initially, all nodes has color red. Start from node $1$, say at round $0$. Node $1$ knows the information of the cut $(S_1, V\setminus S_1)$. To compute conductance of the cut $(S_1, V\setminus S_1)$, node $1$ should know the value of $|E(S_1, V\setminus S_1)|$ (which is its degree in this case) and $\mbox{vol}(S_1)$. In fact, node $1$ knows these value. Node $1$ changes its color as green. 

%\ENDFOR

\STATE Each node $j$ sends the {\em information} which contains their ID, $L^{\pi}_j$ and $R^{\pi}_j$, to  node $1$. 

\STATE Node $1$ computes all the conductance of all $n-1$ cuts locally using this information. %The conductance of the $j$-th cut $(S_j, \bar{S}_j)$ is the ratio between $|E(S_j, \bar{S}_j)|$ and $\min \{\mbox{vol}(S_j), 2m - \mbox{vol}(S_j)\}$. Here, $|E(S_j, \bar{S}_j)| = |E(S_{j-1}, \bar{S}_{j-1})| - L^{\pi}_j + R^{\pi}_j$ and vol$(S_j) = $ vol$(S_{j-1}) + L^{\pi}_j + R^{\pi}_j$ where $|E(S_1, \bar{S}_1)| =$ degree$(1)$ and vol$(1) = $ degree$(1)$.  

\STATE $C^{\ell} \leftarrow$ cut of minimum conductance i.e., $\phi(C^{\ell}) = \min_{j= 1,2, \ldots, n-1} \phi(S_j)$. 
 
\ENDFOR

\STATE Node $1$ chooses the cut $C_{s_k}$ of minimum conductance among all $C^{\ell}$ i.e.,  $\phi(C_{s_k}) = \min_{\ell = 1}^{\log n} \phi(C^{\ell})$. 

\ENDFOR
%\end{algorithmic}

\STATE Node $1$ broadcasts the cut which has minimum conductance among all $C_{s_k}$, to all the nodes in the network. %Say the minimum cut is $C = (C_t, \bar{C}_t)$. Then it is enough for node $1$ to broadcast the ID of  node $t$ only. %He broadcast the value $k$ to all the nodes. Therefore all nodes know the cut with minimum conductance is $C_{\ell} = (S_k, V\setminus S_k)$ for $\ell$-length random walk.  

%\ENDFOR

%\STATE Output the cut $C$ which has minimum conductance among $C_{\ell}$, considering all $\ell$.

\end{algorithmic}
\end{algorithm} 

\subsection{Description and Analysis of the Algorithm}
We describe the algorithm {\sc SparseCut} in detail here. First we want to compute the probability distribution of random walks starting from a source node. It is shown in \cite{SpielmanT04} (cf. Theorem \ref{thm:conductance-estimate}) that the source node should be from the smaller side of the cut of given conductance $\phi$. Since we do not know about the cut set, we cannot choose such a source node. However, if the balance of the cut is $b$, if we choose $\log n/b$ source nodes uniformly at random from $V$, then with high probability at least one node should be from the smaller side of the cut. 
%The proof follows easily from the definition of balance of a cut. 
%Therefore, we run our algorithm for each $\log n/b$ different source node randomly chosen from $V$ to compute a cut with required sparsity.  

For each source node, we compute landing probability distribution of random walks of length $\ell$. The length could be at most $O(1/\phi)$ (cf. Theorem \ref{thm:conductance-estimate}) in the range of $\{1, 2, \ldots, O(1/\phi)\}$. As mentioned earlier, we run our algorithm for $\log n$ different lengths $\ell$, chosen uniformly at random in this range. For simplicity, we  break the remaining portion of the algorithm into two parts: Phase 1 and Phase 2. We run these two phases for each length $\ell$ and for every (chosen) source node. In Phase 1, we  partition the vertex set $V$ according to the prefixes of decreasing order of the ratio of node probability to its degree. First, source node calls the algorithm {\sc EstimateProbability} with input $\ell$ and $K$ to estimate landing probability over nodes. After computing approximate probability distribution $\tilde p(i)$, each node sends the value $\rho(i) = \tilde p(i)/d(i)$ to all other nodes in the network (cf.  proof of the Lemma~\ref{lem:phase1-sc}). Then, we arrange the set of vertices in decreasing order of $\rho(i)$, say, the ordered set is $\pi_{\tilde p} = \{1, 2, \ldots, n\}$. At the end of this phase, each node knows all the partitions $(S_j, \bar{S}_j)$, for all $ j = 1, 2, \ldots, n-1$. In phase 2, we  compute the conductances of these $n-1$ partitions. We describe in Lemma \ref{lem:parallel-conductance} on how to compute conductances of all these cuts in linear time. 

There are $n-1$ partitions (cut sets) corresponding to each $\ell$. Then node $1$ chooses the minimum of the $\Theta(n\log n)$ (among all $\ell$) minimum conductance cuts. Say, the cut is $C_{s_k}$, where $\phi(C_{s_k}) = \min_{\ell} \{\min_{j} \phi(S_j)\}$. Then node $1$ chooses the minimum conductance cut among all source nodes $s_k$ (there are total $\log n/b$) and broadcasts it to all the nodes in the network. Let the minimum cut be $C = (C_t, \bar{C}_t)$, then it is enough for node $1$ to broadcast the node $t$ only as all nodes know the ordered set $\pi_{\tilde p}$. 

\iffalse
\paragraph{Note:} In the algorithm {\sc SparseCut}, we are sending the quantity $\rho(i) = p(i)/d(i)$ of each node $i$ to all other nodes in the network (in phase~1). Then each node knows about the ordering of the vertices and therefore all the partitions. This help us to understand the picture clearly when we are computing $L^{\pi}_j$ and $R^{\pi}_j$  for each node $j$. However, it is easy to see that every node $j$ can compute $L^{\pi}_j$ and $R^{\pi}_j$ without knowing the probability distribution of all other nodes. Each node just find their neighbor's $\rho$ value and check how many of them has less value or greater value than its own $\rho$ value. This number would give the $L^{\pi}_j$ and $R^{\pi}_j$ respectively to the node $j$. Anyway, in both the cases, the overall running time bound of the {\sc SparseCut} algorithm does not change.
\fi  

\subsection{Time Complexity Analysis}
We now analyze the running time of the algorithm {\sc SparseCut}. The following lemmas are required to prove the time complexity of {\sc SparseCut}. 

\begin{lemma}\label{lem:phase1-sc}
Phase~1 of {\sc SparseCut} (cf. Algorithm \ref{alg:sparsecut}) takes $O(\frac{1}{\phi} + n)$ rounds.
\end{lemma}
\begin{proof}
In phase~1, we estimate probability distribution $p(i)$  using the {\sc EstimateProbability} Algorithm. The running time of {\sc EstimateProbability}, following Lemma \ref{lem:time-randomwalk}, is $O(\ell)$ rounds. 
After estimating the landing probability, each vertex sends the quantity $\rho(i) = \tilde p(i)/d(i)$ to all vertices in the network. A simple way of sending these $n$ value to $n$ nodes can be done by constructing a BFS tree (e.g., by first electing a leader). We first construct a BFS tree using the value $\rho(t)$ of each node as its rank. Then the node of highest $\rho$ value would be the root of the tree. Each node upcasts its $\rho$ value to the root node through tree edges. Then the root node floods all $\rho(t)$ to reach all the nodes through the tree edges. It is shown in \cite{peleg} that the upcast and then flooding $n$ values through tree edges can all be done in $O(n+D)$ rounds, where $D$ is the diameter of the graph. Also constructing BFS can be done in $O(D)$ rounds (e.g., \cite{khan-podc}). 
%We note that another way of sending these $n$ different values to all nodes in $O(n + D)$ rounds, is priority based information dissemination\footnote{In each round, every node broadcasts the maximum priority information among all unused (i.e., not broadcasted earlier) information it holds in that round.}. 

All other computations are done locally. Therefore, the total time required for Phase~1 is $O(\ell + n + D)$ rounds. However, the algorithm {\sc EstimateProbability} is called for $\Theta(\log n)$ different random walk lengths, where each length value is at most $O(1/\phi)$.  Also the diameter $D$ is at most $O(n)$ for any graph. Therefore, phase~1 finishes in $O(1/\phi + n)$ rounds.  
\end{proof} 

\begin{lemma}\label{lem:phase2-sc}
Phase~2 of {\sc SparseCut} (cf. Algorithm \ref{alg:sparsecut}) takes $O(n)$ rounds.
\end{lemma}
\begin{proof}
Phase~2 is for computing conductance of $n-1$ cuts $(S_j, \bar{S}_j), j = 1, 2, \ldots, n-1$ where $S_j = \{1, 2, \ldots, j\}$ according to the ordering in $\pi_{\tilde p}$. Therefore, it follows from the proof of Lemma~\ref{lem:parallel-conductance} that node $1$ can compute these $(n-1)$ conductances in $O(n)$ rounds. 
\end{proof}

\begin{theorem}\label{thm:time-sparsecut}
The running time of {\sc SparseCut} (cf. Algorithm \ref{alg:sparsecut}) is $O(\frac{1}{b}(\frac{1}{\phi} + n)\log^2 n)$ rounds where $\phi$ is the conductance of the graph and $b$ is balance of the cut. 
\end{theorem}
\begin{proof}
The algorithm {\sc SparseCut} essentially runs in two phases inside first two {\bf for} loops, one is for choosing source nodes and other is for choosing length of random walks. Then at the end, node $1$ performs some local computation to choose the minimum conductance cut and sends it to all other nodes. Sending this to all the nodes in the network can be done in $O(D)$ rounds, which follows from the above discussion of the algorithm. Now, for phase~1 and phase~2, we already have calculated the running time (cf. Lemma \ref{lem:phase1-sc} and Lemma \ref{lem:phase2-sc}). Therefore, adding all these time together, we get $O(\frac{\log n}{b}(1/\phi + n)\log n + D)$ rounds, where the factor $\frac{\log n}{b}$ is for the first {\bf for} loop, the factor $\log n$ for the second {\bf for} loop and last $D$ is for sending the cut information to all nodes. All other computations are dominated by this bound. Since $D$ is dominated by $n$, therefore the running time of the algorithm {\sc SparseCut} reduces to $O(\frac{1}{b}(1/\phi + n)\log^2 n)$ rounds.  
\end{proof}

Combining the above running time lemmas, we prove the main result of this section --- Theorem \ref{thm:algo1} (cf. Section \ref{sec:results}).  

\iffalse
\begin{theorem}\label{thm:main-result}
Given any $n$-node graph $G$ and a conductance $\phi$, there is an algorithm {\sc SparseCut} (cf. Algorithm \ref{alg:sparsecut}) that outputs a cut of conductance at most $\tilde O(\sqrt{\phi})$ with high probability and runs in $\tilde O(\frac{1}{b}(\frac{1}{\phi} + n))$ rounds, where $\phi$ is the conductance of the graph and $b$ is the balance of the optimal cut.
\end{theorem}
\fi

\begin{proof} (of Theorem \ref{thm:algo1})
The approximation guarantee of algorithm {\sc SparseCut}, i.e., it computes a cut with sparsity $\tilde O(\sqrt{\phi})$ follows from  Theorem~\ref{thm:conductance-estimate}. We choose $\eps = O(\phi^2)$. Moreover, we are performing random walks up to length $O(1/\phi)$. Therefore, it follows from  Theorem~\ref{thm:conductance-estimate} that our algorithm computes a cut with conductance $\tilde O(\sqrt{\phi})$. 
The running time of the algorithm follows from the above Theorem~\ref{thm:time-sparsecut} which is $\tilde O(\frac{1}{b}(\frac{1}{\phi} + n))$ rounds. 

In the {\sc SparseCut} algorithm, we are required to compute probability distributions by performing random walks from a {\em good} source node to satisfy the condition of  Theorem \ref{thm:conductance-estimate}. A source node is {\em good} if it is from the smaller side of a desired cut (as shown in \cite{SpielmanT04}). 
%That's why we are running our algorithm for $\log n/b$ source nodes, where $b$ is the balance of the cut. 
If we are interested in finding a cut of constant balance, then $b$ is constant. Therefore, as an immediate corollary,  computing a sparse cut of constant balance 
takes $\tilde O(\frac{1}{\phi} + n)$ rounds.
\end{proof}

\iffalse
\begin{corollary}\label{cor:constant-balance}
Given any graph $G$ that has a cut of constant balance and conductance (at most) $\phi$, there is an algorithm (cf. Algorithm \ref{alg:sparsecut}) that outputs a cut set of conductance at most $\tilde O(\sqrt{\phi})$ with high probability and runs in $\tilde O(\frac{1}{\phi} + n)$ rounds. 
\end{corollary}   
 \fi

The analysis of our algorithm is tight.  Consider the barbell graph $B_n$ which is a graph consisting of two cliques of size $(n-1)/2$ connected by a path of length $2$ (see, figure $2$ in \cite{AlonAKKLT11}). Consider a source node $s$ in one clique. Then to compute the smallest conductance cut (one set of which would be the clique containing $s$), the random walk starting from $s$, should reach the second clique.  This will take at least $\Theta(n^2)$ rounds, which is bounded by $\Omega(1/\phi)$. Then to collect all the  information as in Lemma \ref{lem:parallel-conductance} at the node $s$ will take $\Omega(n)$ rounds. Hence, total time required is $\Omega(1/\phi + n)$.

\section{Finding Local Cluster Set}\label{sec:local-cluster}
We describe an approach to compute a local cluster, i.e., a subset of vertices containing  a given source node $v$ such that the internal edge connections are significantly higher than the outgoing edges from it. 

Suppose a source node $s \in V$ is given. First, guess a conductance $\phi$ starting from a constant (say $1/2$, which is essentially the best possible) and then run the above {\sc SparseCut} algorithm for the particular node $s$, i.e.,  run the algorithm from Step 3 for source node $s$. Then check whether the smallest conductance satisfies the quadratic factor approximation. If yes, we stop; otherwise, we  halve the (guessed) conductance and continue. Since the minimum conductance value is $O(1/m)$, we need to do at most $O(\log n)$ guesses, as $m = O(n^2)$. The running time bound of the algorithm for computing a local cluster is stated in Theorem \ref{thm:cluster} (cf. Section \ref{sec:results}) and the proof is given below.  

%\begin{theorem}\label{thm:cluster}
%Given an $n$-node network $G$ and  source node $s$, there is a  distributed algorithm that  outputs a {\em local} cluster in $\tilde O(\frac{1}{\phi} + n)$ rounds, where $\phi$ is the conductance of the graph. 
%\end{theorem}
\begin{proof} (of Theorem \ref{thm:cluster})
We run the {\sc SparseCut} algorithm only for one specified source node. The running time of {\sc SparseCut} algorithm for a single source node is $\tilde O(1/\phi + n)$ rounds with high probability. Checking whether the smallest conductance satisfies  the quadratic factor approximation can be done locally at the source node $s$. Then we may have to run the algorithm at most $O(\log n)$ times for guessing the (best possible) conductance.  Therefore, the running time of the algorithm is $\tilde O(1/\phi + n)$ rounds with high probability. 
\end{proof}

\section{Sparse Cuts using PageRank}\label{sec:pagerank-algo}

In this section, we present another approach to compute a sparse cut of an undirected graph $G = (V, E)$. This is a variant of the first algorithm and based on PageRank computation. We derive an algorithm following \cite{AndersenCL06} and adapt it to the  CONGEST distributed computing model and obtain similar guarantees as before, i.e., a quadratic approximation.
% However, this approach can also be extended to works for finding global cut of a graph.    

Recall that in the previous section we use random walk probability distributions to find candidate partitions of the vertex set. Now instead of standard random walk, we use another well known distribution vector called {\em PageRank} to partition vertices. The PageRank of a graph (e.g., \cite{page99,anatomy+page98}) is the {\em stationary distribution} vector of the following special type of random walk: at each step of the walk, with some probability $\alpha$ it starts from a randomly chosen node and with remaining probability $1-\alpha$, it follows a randomly chosen neighbor from the current node and moves to that neighbor. The parameter $\alpha$ is called {\em reset} probability or {\em teleport} probability. The personalized PageRank (e.g., \cite{ppr-bahmani2010} and references therein) is the stationary distribution vector of a slightly modified random walk as of PageRank: In every round, instead of starting from a randomly chosen node with probability $\alpha$, the walk restarts from the source node itself and with remaining probability $1-\alpha$, the walk moves to a random neighbor from the current node. This alternative approach of graph partitioning, based on personalized PageRank vectors, was studied by Andersen et al. in \cite{AndersenCL06} in centralized setting. They show an improved result similar to Spielman et al. \cite{SpielmanT04} using personalized PageRank vectors with better approximation and running time. In this paper, we build on the results of \cite{AndersenCL06} and present a distributed algorithm to compute sparse cuts. Along the way, we also present a simple and efficient distributed algorithm to compute personalized PageRank. Throughout this section, by random walk we mean this special type of random walk unless otherwise stated. 

We next discuss estimation of personalized PageRank vectors in the distributed CONGEST model.   
First we introduce some notation.
Let ${\bf p}(q)$ denote the PageRank vector with respect to a given {\em starting} vector $q$, i.e., the starting node is chosen with distribution $q$.
The personalized PageRank vector with respect to a given node $v$ can be denoted by ${\bf p}(\chi_v)$, where the starting vector $\chi_v$ is the characteristic vector of $v$ (i.e., it is 1 at $v$'s coordinate and 0 elsewhere). We compute an $\eps$-approximate PageRank vector ${\bf \tilde p}(\chi_v)$ which is within
an additive error of $\eps$.  For technical reasons (cf. Section \ref{sec:loccut}), we take $\eps$ be $O(1/n^4)$.

\subsection{Estimating Personalized PageRank}\label{sec:pagerank-estimate} %$p_{\ell}(s, i)$} 
We derive a simple approach to estimate the personalized PageRank vector. We present a Monte Carlo based distributed algorithm for computing personalized PageRank of a graph, similar to \cite{DasSarmaMPU13}. The main idea is as follows. Perform many random walks starting from a specific source node $s$. In every round, each random walk independently  goes to a random neighbor with probability $1-\alpha$ and with the remaining probability (i.e., $\alpha$) terminates in the current node. We note that the random walk here means the personalized PageRank random walk. Since, $\alpha$ is the probability of termination of a walk in each round, the expected length of every walk is $1/\alpha$ and the length will be at most $O(\log n/\alpha)$ with high probability. During this process, every node $v$ counts the number of visits (say, $\eta_v$) of all the walks that go through it. Suppose the number of random walks starting from $s$ is $K$. Then, after termination of all walks in this process, each node $v$ computes (estimates) its personalized PageRank $p(v)$ as $\tilde p(v) = \frac{\eta_v \alpha}{K}$. Notice that $\frac{K}{\alpha}$ is the (expected) total number of visits over all $n$ nodes of all the $K$ walks. The above idea of counting the number of visits is a standard technique to approximate PageRank (see e.g., \cite{mcm-avrachenkov,ppr-bahmani2010}). We first present the algorithm in a pseudocode (cf. Algorithm \ref{alg:pr-walk}) to approximate $p(v)$ and then analyze the result on accuracy of estimation below.     

\begin{algorithm}[!h]
\caption{\sc EstimatePageRank}
\label{alg:pr-walk}
\textbf{Input:} Source node $s$, reset probability $\alpha$, and number of walks $K$.\\% = \Theta(n \log n/\alpha)$.\\
\textbf{Output:} Approximate PageRank $\tilde{p}(v)$ of each node $v$.\\
\begin{algorithmic}[1]

\STATE Node $s$ floods the value $K = n^4\log n$, the number of random walks to be performed to all other nodes.

\STATE Source node $s$ creates $K$ random walk tokens and performs these simultaneously. All walks keep moving in parallel until they TERMINATE. 

\STATE Every node maintains a counter number $\eta_v$ for counting visits of random walks to it. 

\WHILE{there is at least one (alive) token}

\STATE This is $i$-th round. Each node $v$ holding at least one token does the following: Consider each random walk token $\mathcal{C}$ held by $v$ which is received in the $(i-1)$-th round. Generate a random number $r \in [0, 1]$.

\IF{$r< \alpha$} 
\STATE Terminate the token $\mathcal{C}$.
\ELSE
\STATE Select an outgoing neighbor uniformly at random, say $u$. Add one token counter number to $T^v_u$ where the variable $T^v_u$ indicates the number of tokens (or random walks) chosen to move to the neighbor $u$ from $v$ in the $i$-th round.    
\ENDIF

\STATE Send the token's counter number $T^v_u$ to the respective outgoing neighbor $u$. 

\STATE Every node $u$ adds the total counter number ($\sum_{v \in N(u)} T^v_u$---which is  the total number of visits of random walks to $u$ in $i$-th round) to $\eta_u$.

\ENDWHILE

\STATE Each node outputs its personalized PageRank as $\frac{\eta_v \alpha}{K}$.

\end{algorithmic}

\end{algorithm}  

\subsubsection{Analysis}\label{subsec:analysis}

Now we show that the algorithm {\sc EstimatePageRank} (cf. Algorithm \ref{alg:pr-walk}) gives an estimation $\tilde p(v)$ of $p(v)$ with very high accuracy. The algorithm outputs the personalized PageRank of each node $v$ as $\tilde p(v) = \frac{\eta_v \alpha}{K}$. The correctness of the above approximation follows directly from the analysis of the Algorithm~1 in \cite{DasSarmaMPU13}. However, the algorithm of \cite{DasSarmaMPU13} is for computing the general PageRank (not personalized) (using  an approach due  to  \cite{mcm-avrachenkov}). However, it is easy to verify that the approach is equivalent for both general PageRank and personalized PageRank. This is because in general PageRank computation \cite{DasSarmaMPU13},  several random walks are performed from every node and the walks are terminated with reset probability (instead of restarting from a random node). Now for personalized PageRank, we perform several random walks from a {\em particular source node} and terminate each walk with reset probability (instead of restarting from the source node again). Therefore, in both cases, the random walks start again independently with probability $\alpha$ from source node(s). Hence, our approach also correctly outputs the personalized PageRank vector.   

It is shown in \cite{DasSarmaMPU13} that by performing total $\Theta(\log n)$ random walks from each node, we get a sharp approximation of PageRank vector with high probability. Therefore,  for personalized PageRank, it is enough   to get a good accuracy, if we perform $K = \Theta(n \log n)$ random walks from a particular source node. However, we can perform much more walks to get very high accuracy as needed here. In particular, we show later that it would be sufficient for our algorithm to perform $O(n^4 \log n)$ random walks. Below is a lemma on the running time of our algorithm.  

\begin{lemma}\label{lem:time-pr-walk}
Algorithm {\sc EstimatePageRank} (cf. Algorithm \ref{alg:pr-walk}) computes personalized PageRank in $\tilde O(\frac{1}{\alpha})$ rounds with high probability, where $\alpha$ is the reset probability.   
\end{lemma}
\begin{proof}
To prove the lemma, we first show that there is no congestion in the network if the source node starts at most a polynomial (in $n$) number of random walks simultaneously. This is because, nodes are only sending the `count' number of random walk tokens in the algorithm. The process is similar to that in Section~\ref{sec:prob-estimate} where we estimate the landing probability distribution using the same technique. Hence the claim on the congestion part follows from the proof of the Lemma \ref{lem:time-randomwalk}. 

Now it is clear that the algorithm stops when all the walks terminate. Since the termination probability is $\alpha$, so in expectation after $1/\alpha$ steps, a walk terminates and with high probability (via the Chernoff bound) the walk terminates in $O(\log n/\alpha)$ rounds; by union bound \cite{MU-book-05}, all walks (since they are only polynomially many) terminate in $O(\log n/\alpha)$ rounds with high probability as well. Since all the walks are moving in parallel and there is no congestion, all the walks in the graph terminate in $O(\log n/\alpha)$ rounds whp. 
\end{proof} 

\subsection{Algorithm for Sparse Cut using PageRank}\label{sec:loccut}
We describe an algorithm to compute a sparse cut in $G$.
The idea is very similar to the previous section (cf. Algorithm \ref{alg:sparsecut}). In the previous section we used standard random walk to find the partitions of vertex set. Here we use personalized PageRank for partitioning and arrange the vertices in decreasing order of the ratio: (PageRank)/(degree of vertex). Consider $(n-1)$ partitions according to this ordering and compute conductance for each of them. Then the cut of minimum conductance is at most $\tilde O( \sqrt{\alpha})$, if we performed random walk from a specified vertex with reset probability $O(\alpha)$ (we will take $\alpha$ to be $\Theta(\phi)$). This guarantee follows from the result of \cite{AndersenCL06} stated below (in modified form for our purposes). 
%

%The above procedure always finds a cut of minimum conductance which is at most $\tilde O( \sqrt{\alpha})$, if we compute PageRank from a specified vertex with reset probability $O(\alpha)$. The guarantee follows from a result of \cite{AndersenCL06}, stated below (in modified form). 
%We note that in this section, we refer to a partition $(C, \bar{C})$ sometimes by just a cut $C$.  

\begin{theorem}[\cite{AndersenCL06}]\label{thm:pr-technical-result}
Let $C_v$ be a cut containing node $v$ with a conductance $\phi$. If $\tilde{p}$ is an $\eps$-approximation to the personalized PageRank vector ${\bf p}(\chi_v)$ \footnote{Actually,  the result holds
for a slightly different type of approximate PageRank vector defined in \cite{AndersenCL06}; nevertheless, this can be shown to be closely approximated by  $\eps$-approximate PageRank vector $\tilde{p}$ as defined here, if we choose $\eps$ small enough, i.e., $\eps = O(1/\mbox{vol}(C_v)^2)$, i.e., $O(1/n^4)$.}  computed with the reset probability $\alpha = 10\phi$,  and $\eps = O(\frac{1}{n^4})$, then the smallest conductance of $n-1$ cuts according to the ordering of vertices in decreasing order of $\tilde p(i)/d(i)$ is $\phi(\tilde p) = \tilde O(\sqrt{\phi}).$
\end{theorem}

\paragraph{Algorithm:} Similar to {\sc SparseCut} algorithm in Section \ref{sec:sparse-cut}. First we have to choose a {\em good} source node, i.e., a node $s$ from the smaller side of the cut of given conductance $\phi$. For this, we choose $O(\log n/b)$ uniformly random nodes, assuming the balance $b$ of the cut is given. This will guarantee that at least one node is from the smaller side with high probability. Then do the following for every source node $s$: compute the personalized PageRank vector using algorithm {\sc EstimatePagerank} with source node $s$ and reset probability $\alpha = 10 \phi$. Compute conductances of $(n-1)$ cuts derived from the PageRank vector as explained above. Then output the cut set with minimum conductance among all $(n - 1)\cdot \frac{\log n}{b}$ cuts. Notice that there are $(n - 1)$ cuts for one source node. Then the output cut would have sparsity $\tilde O(\sqrt{\phi})$ which follows from the Theorem~\ref{thm:pr-technical-result}. The reset probability $\alpha$ of the PageRank is chosen $10\phi$ according to the theorem \ref{thm:pr-technical-result}. The following theorem state the main result of this section.   
 
 \begin{theorem}\label{thm:pr-sparsecut}
Given any graph $G$ and a conductance at most $\phi$, there is a PageRank based algorithm that computes a cut set of conductance at most $\tilde O(\sqrt{\phi})$ with high probability in $\tilde O(\frac{1}{b}(\frac{1}{\phi} + n))$ rounds, where $b$ is the balance of the cut.  
\end{theorem}
\begin{proof}
The algorithm runs in two phases for each of $O(\log n/b)$ source nodes. The first phase is for computing personalized PageRank and it takes $\tilde O(\frac{\log n}{\phi})$ rounds with high probability (cf. Lemma \ref{lem:time-pr-walk}). The second phase is similar to the Algorithm \ref{alg:sparsecut}. That is, computing partitions according to the PageRank,  and then computing conductances of all partitions: all this can be done in $\tilde O(n + D)$ rounds. Hence totally we have $\tilde O(1/\phi + n)$ rounds with high probability, since diameter $D \leq n$. Therefore, over all the $O(\log n/b)$ source nodes, the running time of the PageRank based algorithm is $\tilde O(\frac{1}{b}(\frac{1}{\phi} + n))$ rounds with high probability.    
\end{proof}

\section{Lower Bound}\label{sec:lower-bound} 
We derive a general lower bound for the distributed sparse cut problem. In particular, we show that there is graph in which any approximation algorithm for computing sparsest cut will take $\tilde \Omega(\sqrt{n} + D)$ rounds, where $D$ is the diameter of the graph. We use the technique of  \cite{DasSarmaHKKNPPW12} which shows almost tight lower bounds for many distributed verification and optimization problems. Their lower bound proofs rely on a bridge between communication complexity and distributed computing.
%We use their results and similar technique to show the lower bound of sparsest cut problem. 

We show a reduction from the spanning connected subgraph {\em verification problem} to the sparsest cut (optimization) problem. In the spanning connected  subgraph verification problem, given a graph $G = (V, E)$ and a  subgraph $H = (V, E')$ with $E' \subseteq E$, it is required to check whether the subgraph $H$ is a spanning connected subgraph of $G$ via a distributed algorithm. We convert the spanning connected subgraph verification problem to the sparsest cut  problem (with edge weights). In particular, we show that an $\alpha$-approximation $\epsilon$-error algorithm\footnote{A randomized algorithm $\mathcal{A}$ is $\alpha$-approximation $\epsilon$-error if for any input , the algorithm $\mathcal{A}$ outputs a solution that is at most $\alpha$ times the optimal solution of the input with probability at least $1-\epsilon$.} $\mathcal{A}$ for sparsest cut problem, can be used to solve the spanning connected subgraph verification problem using the same running time. Hence the lower bound proved in \cite{DasSarmaHKKNPPW12} (cf. Theorem 5.1) for the spanning connected subgraph verification problem (which is $\tilde \Omega(\sqrt{n} + D)$), also applies to the sparsest cut computation problem. We use the   graph $G(\Gamma, d, p)$  (this graph with parameters $\Gamma, d,$ and $p$ is defined in \cite{DasSarmaHKKNPPW12}) to show the lower bounds. We consider the same parametrized graph $G = G(\Gamma, d, p)$, which is connected by our assumption.  The reduction from the spanning connected  subgraph verification problem is direct:  In  $G$ we assign a  weight of 1 to all edges in the subgraph $H$ and weight $0$ to all other edges. Now, observe that if $H$ is not connected then the conductance of sparsest cut is $0$, since we can then partition the whole graph into two components and all the edges crossing the two components has weight $0$. On the other hand, if $H$ is connected then every cut set contains at least one edge from $H$, which implies that the conductance of the sparsest cut would be 
non-zero. Thus, any algorithm with non-trivial approximation ratio will be able to distinguish the two cases.  

Therefore, it follows that the sparsest cut computation problem has a lower bound $\tilde \Omega(\sqrt{n} + D)$.

\section{Conclusion}\label{sec:conclusion}
We presented distributed approximation algorithms for computing sparse cuts,
with provable guarantees on the conductance. For future work, one can try to improve the running time bound $\tilde O(\frac{1}{\phi} + n)$ rounds. There is previous work on performing an $\ell$ length random walk in time $\tilde O(\sqrt{\ell D)}$ rounds \cite{drw-jacm}. This can be used to potentially speed up random walks and hence reduce the ``$\frac{1}{\phi}$ part" of the time bound, since walks of that much length has to be performed. (As mentioned earlier, since $1/\phi = \Omega(D)$, this cannot be improved beyond $\Omega(D)$ because of our lower bound of $\Omega(D + \sqrt{n})$.) However, the technique in \cite{drw-jacm} may not be applicable directly here because of congestion; we need to perform many random walks to compute the landing probability distribution with high enough accuracy. One might also try to improve the ``$n$" part of the time bound  and see if we can match the $\Omega(\sqrt{n})$ lower bound. Improving this seems to depend on computing the conductance of $n-1$ different cuts in time that is sublinear in $n$, which seems harder; alternatively it may be possible to try significantly fewer than $n$ cuts in each of our distributional orders and still guarantee an approximation bound.

Our sparse cuts computation can be used to identify the crossing edges, which have been used in prior work (\cite{mihail}) to heuristically improve network search, routing, and connectivity. It will be useful to rigorously show such results with provable guarantees. 

%\newpage

  \let\oldthebibliography=\thebibliography
  \let\endoldthebibliography=\endthebibliography
  \renewenvironment{thebibliography}[1]{%
    \begin{oldthebibliography}{#1}%
      \setlength{\parskip}{0ex}%
      \setlength{\itemsep}{0ex}%
  }%
  {%
    \end{oldthebibliography}%
  }
%{ \small
{%\tiny
\bibliographystyle{abbrv}
\bibliography{Distributed-RW}
}

\end{document}